%% file: template_isit20.tex
\documentclass[journal,letterpaper, onecolumn, draftclsnofoot, 12pt]{IEEEtran}
\IEEEoverridecommandlockouts

\usepackage[utf8]{inputenc} 
\usepackage[T1]{fontenc}
\usepackage{url}
\usepackage{ifthen}
\usepackage{cite}
\usepackage[cmex10]{amsmath} 
                             
\usepackage{pgfplots}
\DeclareUnicodeCharacter{2212}{−}
\usepgfplotslibrary{groupplots,dateplot}
\usetikzlibrary{patterns,shapes.arrows}
\pgfplotsset{compat=newest}

\usepackage{amsfonts}
\usepackage{amssymb}
\usepackage{mathtools}
\usepackage{enumitem}
\usepackage{tikz}
\usepackage{pgfgantt}
\usepackage{pdflscape}
\usepackage{prettyref}
\usepackage{refstyle}
\usepackage{amsthm}
\pgfplotsset{compat=1.3} 
\pgfplotsset{plot coordinates/math parser=false}
\usepackage{subfig}
\makeatletter

\theoremstyle{definition}

\theoremstyle{plain}

\theoremstyle{definition}
\newtheorem{defn}{\protect\definitionname}
\theoremstyle{plain}
\newtheorem{thm}{\protect\theoremname}
\theoremstyle{plain}
\newtheorem{lem}{\protect\lemmaname}
\ifx\proof\undefined\
  \newenvironment{proof}[1][\proofname]{\par
    \normalfont\topsep6\p@\@plus6\p@\relax
    \trivlist
    \itemindent\parindent
    \item[\hskip\labelsep
          \scshape
      #1]\ignorespaces
  }{%
    \endtrivlist\@endpefalse
  }
  \providecommand{\proofname}{Proof}
\fi

\newref{lem}{name=Lemma~}
\newref{thm}{name=Theorem~}
\newref{exa}{name=Example~}
\newref{sec}{name=Section~}
\newref{def}{name=Definition~}
\newref{fig}{name=Fig.~}
\newref{subsec}{name=Subsection~}
\newref{rem}{name=Remark~}
\newref{eq}{refcmd=(\ref{#1})}

\providecommand{\definitionname}{Definition}
\providecommand{\examplename}{Example}
\providecommand{\lemmaname}{Lemma}
\providecommand{\remarkname}{Remark}
\providecommand{\theoremname}{Theorem}

\newlength\fwidth

\begin{document}
\bstctlcite{IEEEexample:BSTcontrol}
\title{Speeding Up Private Distributed Matrix Multiplication via Bivariate Polynomial Codes} 

\author{%
  \IEEEauthorblockN{Burak Hasırcıoğlu\IEEEauthorrefmark{1},
                    Jesús Gómez-Vilardebó\IEEEauthorrefmark{2},
                    and Deniz Gündüz\IEEEauthorrefmark{1}}
                    
  \IEEEauthorblockA{\IEEEauthorrefmark{1}%
                    Imperial College London,
                    UK,
                    \{b.hasircioglu18, d.gunduz\}@imperial.ac.uk}
  \IEEEauthorblockA{\IEEEauthorrefmark{2}%
                    Centre Tecnològic de Telecomunicacions
de Catalunya (CTTC/CERCA), Barcelona, Spain,
                    jesus.gomez@cttc.es}
}

\maketitle

\begin{abstract}
We consider the problem of private distributed matrix multiplication under limited resources. Coded computation has been shown to be an effective solution in distributed matrix multiplication, both providing privacy against the workers and boosting the computation speed by efficiently mitigating stragglers. In this work, we propose the use of recently-introduced bivariate polynomial codes to further speed up private distributed matrix multiplication by exploiting the partial work done by the stragglers rather than completely ignoring them. We show that the proposed approach reduces the average computation time of private distributed matrix multiplication compared to its competitors in the literature while  improving the upload communication cost and the workers' storage efficiency.
\end{abstract}



\section{Introduction}\label{sec:intro}

Matrix multiplication is a fundamental building block of many applications in signal processing and machine learning. For some applications, especially those involving massive matrices and 
stringent latency requirements, matrix multiplication in a single computer is infeasible,
and distributed solutions need to be adopted. In such a scenario, the full multiplication task is first partitioned into smaller sub-tasks, which are then distributed across dedicated \emph{workers}. 

In this work, we address two main challenges in distributed matrix multiplication. The first one is the \emph{straggler}s, which refers to unresponsive or very slow workers. If the completion of the full task requires the computations from all the workers, straggling workers become a significant bottleneck. To compensate for the stragglers, additional redundant computations can be assigned to workers. It has been recently shown that the use of error-correcting codes, by treating the slowest workers as erasures, instead of simply replicating tasks across workers, significantly lowers the overall computation time \cite{lee2017speeding}. 

In the context of straggler mitigation, polynomial-type codes are studied in \cite{yu_polynomial_2017, dutta_optimal_2019, yu_straggler_2018-1, jia2019cross}. In these schemes, matrices are first partitioned and encoded using polynomial codes at the master server. Then, workers compute sub-products by multiplying these coded partitions and send the results back to the master for decoding. The minimum number of sub-tasks required to decode the result is referred to as the \emph{recovery threshold}, and denoted by $R_{th}$. In these schemes, only one sub-product is assigned to each worker, and therefore, any work done by the workers beyond the fastest $R_{th}$ is completely ignored. This is sub-optimal since the workers may have similar computational speeds, in which case most of the work done is lost. 



In the \emph{multi-message approach} \cite{kiani_exploitation_2018, amiri_computation_2018, ozfatura2020straggler, hasircioglu2020isit, hasircioglu_globecom_2020}, multiple sub-products are assigned to each worker and the result of each sub-product is communicated to the master as soon as it is completed. This results in faster completion of the full-computation as it allows to exploit partial computations completed by stragglers. Moreover, the multi-message approach makes finishing the task possible even if there are not as many available workers as the recovery threshold.

Another important parameter in distributed computation is the \emph{upload cost}, which is defined as the number of bits sent from the master to each worker, or equivalently, the storage required per worker. 
As discussed in \cite{hasircioglu2020isit}, simply assigning multiple sub-products to the workers using polynomial-type codes is not efficient in terms of the upload cost, as one matrix partition can only be used in the computation of one sub-product. To combat this limitation, product codes are considered in \cite{kiani_exploitation_2018} for the multi-message distributed matrix multiplication problem. However, with product codes, sub-products are no longer fully one-to-any replaceable, which reduces the scheme's resource efficiency. To make the sub-products one-to-any replaceable, \textit{bivariate polynomial codes} are introduced in \cite{hasircioglu2020isit, hasircioglu_globecom_2020}, which provides a better trade-off between the upload cost and expected
computation time, by allowing a matrix partition to be used in the computation of several sub-tasks.

The second challenge we tackle in this paper is privacy. The multiplied matrices may contain sensitive information, and sharing these matrices even partially with the workers may cause a privacy breach. Moreover, in some settings, a number of workers can exchange information with each other to learn about the multiplied matrices. Such a collusion may result in leakage even if no information is revealed to individual workers. The first application of polynomial codes to privacy-preserving distributed matrix multiplication is presented in \cite{chang2018capacity}. To hide the matrices from the workers, random matrix partitions are created, and linearly encoded together with the true matrix partitions using polynomial codes. The recovery threshold has been improved in subsequent works \cite{kakar2018rate}, \cite{d2020gasp}, by carefully choosing the degrees of the encoding monomials so that the resultant decoding polynomial contains the minimum number of additional coefficients. In \cite{aliasgari2020private, jia2019capacity,kakar2019uplink}, lower recovery threshold values than \cite{d2020gasp} are obtained by using different matrix partitioning techniques and different choices of encoding polynomials, but this is achieved at the expense of a considerable increase in the upload cost. In \cite{mital2020secure}, a novel coding approach for distributed matrix multiplication is proposed based on polynomial evaluation at the roots of unity in a finite field. It has constant time decoding complexity and a low recovery threshold compared to traditional polynomial-type coding approaches, but the sub-tasks are not one-to-any replaceable and its straggler mitigation capability is limited. In \cite{bitar2020rateless}, a multi-message approach is proposed for private distributed matrix multiplication by using rateless codes. Computations are assigned adaptively in rounds, and in each round, workers are classified into clusters depending on their computation speeds. Results from a worker in a cluster are useful for decoding only if the results of all the sub-tasks assigned to that cluster and also to the fastest cluster are collected, making computations not one-to-any replaceable.


In this work, we propose a multi-message, straggler-resistant, private distributed matrix multiplication scheme based on bivariate Hermitian polynomial codes. Our scheme works effectively even with a small number of workers and under a limited upload cost budget. We show that, especially when the number of fast workers is limited, our proposed method  outperforms other schemes in the literature in terms of the average computation time under a given upload cost budget. We also show that our scheme retains its low expected computation time under both homogeneous and heterogeneous computation speeds across the workers.
\vspace{-5pt}
\section{Problem Setting\label{sec:Problem-Setting}}
\vspace{-5pt}

We study distributed matrix multiplication with strict privacy
requirements. The elements of our matrices are in a finite field $\mathbb{F}_{q}$,
where $q$ is a prime number determining the size of the finite field. There is a master node that can access to the statistically independent matrices $A\in\mathbb{F}_{q}^{r\times s}$ and $B\in\mathbb{F}_{q}^{s\times c}$, $r,s,c\in\mathbb{Z}^{+}$.
The master offloads the multiplication of matrices $A$ and $B$ to $N$ workers, which possibly
have heterogeneous computation speeds and storage capacities. We do not assume any statistics are known about the computation speeds of the workers. To offload the computation to several workers, the master divides the full
multiplication task into smaller sub-tasks and then collects
the responses from the workers. To define these sub-tasks,
the master partitions $A$ into $K$ sub-matrices as $A=\begin{bmatrix}A_{1}^{T} & A_{2}^{T} & \cdots & A_{K}^{T}\end{bmatrix}^{T}$,
where $A_{i}\in\mathbb{F}_{q}^{\frac{r}{K}\times s}$, $\forall i\in[K]\triangleq\{1,2,\dots,K\}$,
and $B$ into $L$ sub-matrices as $B=\begin{bmatrix}B_{1} & B_{2} & \cdots & B_{L}\end{bmatrix}$,
where $B_{j}\in\mathbb{F}_{q}^{s\times\frac{c}{L}}$, $\forall j\in[L]$. The master sends coded versions, i.e., linear combinations, of these partitions to the workers. Assuming worker $i$ can store $m_{A,i}$ partitions of $A$ and $m_{B,i}$ partitions of $B$, the master sends coded partitions $\tilde{A}_{i,k}$ and $\tilde{B}_{i,l}$ to worker $i$, where $i\in[N]$, $k\in[m_{A,i}]$ and $l\in [m_{B,i}]$. For simplicity, we describe a static setting, in which all the coded matrices are sent to the workers before they start computations. However, in a more dynamic setting, in which matrix partitions are sent when they are needed, the required memory at workers could be made smaller. The basic assumption is that $m_{A,i}$ partitions of A and $m_{B,i}$ partitions of B are available for worker $i$ at some point, and thus they could exploit them to extract information on the original matrices $A$ and $B$. The workers multiply the received coded partitions of $A$ and $B$ in a way depending on the underlying coding scheme and send the result of each computation to the master as soon as it is finished. Once the master receives a number of computations equal to the recovery threshold, $R_{th}$, it can decode
the desired multiplication $AB$.

In our threat model, all the workers are honest but curious. That is, they follow the protocol but they can use the received encoded matrices to gain information about the original matrices, $A$ and $B$. We also assume that any $T$ workers can collude, i.e., exchange information among themselves. Our privacy requirement is such that no $T$ workers are allowed to gain any information about the content of the multiplied matrices in the information-theoretic sense. 
\vspace{-5pt} 
\section{Proposed Scheme}
Our coding scheme is based on bivariate polynomial codes \cite{hasircioglu2020isit,hasircioglu_globecom_2020}. Thanks to their lower upload cost, bivariate polynomial codes allow workers to complete more sub-tasks compared to their univariate counterparts under the same upload cost budget, which, in turn, improves the expected computation time and helps to satisfy the privacy requirements.

\subsection{Encoding}

In the proposed coding scheme, coded matrices are generated by evaluating the following polynomials and their derivatives: 
\begin{equation}
\label{eq:a_x}
A(x)=\sum_{i=1}^{K}A_{i}x^{i-1}+\sum_{i=1}^{T}R_{i}x^{K+i-1},    
\end{equation}
\begin{equation}
\label{eq:b_xy}
B(x,y)=\sum_{i=1}^{L}B_{i}y^{i-1}+\sum_{i=1}^{T}\sum_{j=1}^{m}S_{i,j}x^{K+i-1}y^{j-1},
\end{equation}
where $m \leq L$ is the maximum number of sub-tasks any worker can complete. Matrices $R_i\in \mathbb{F}_q^{\frac{r}{K}\times s}$ and $S_{i,j}\in \mathbb{F}_q^{s \times \frac{c}{L}}$ are independent and uniform randomly generated from their corresponding domain for $i\in [T]$ and $j \in [m]$. For each worker $i$, the master evaluates $A(x)$ at $x_{i}$ and the derivatives of $B(x,y)$ with respect to $y$ up to the
order $[m-1]$ at $(x_{i},y_{i})$. We only require these evaluation points to be distinct. Thus, the master sends to worker
$i$, $A(x_{i})$ and $\mathcal{B}_{i}=\{B(x_{i},y_{i}),\partial_{1}B(x_{i},y_{i}),\dots,\partial_{m-1}B(x_{i},y_{i})\}$,
where $\partial_{i}$ is the $i^{th}$ partial derivative with respect
to $y$. Thus, we require $m_{A,i}=1$ and $m_{B,i}=m$.

In \eqref{a_x} and \eqref{b_xy}, the role of $R_i$'s and $S_{i,j}$'s is to mask the actual matrix partitions for privacy. The following theorem states that the evaluations of $A(x)$, $B(x,y)$ and its derivatives do not leak any information about $A$ and $B$ to any $T$ colluding workers.

\begin{thm}
For the encoding scheme described above, we have
\begin{equation}
I(A,B;\{(A(x_i),\mathcal{B}_i)\mid i\in \tilde{N}\})=0,    
\end{equation}
$\forall \tilde{N}\subset [N]$ such that $|\tilde{N}|=T$.
\end{thm}

\begin{proof}
Since $A$ and $B$ are independent, we have
\begin{equation*}
I(A,B;\{(A(x_{i}),\mathcal{B}_{i})\mid i\in \tilde{N}\})=I(A;\{A(x_{i})\mid i\in \tilde{N}\})+I(B;\{\mathcal{B}_{i}\mid i\in \tilde{N}\}).
\end{equation*}
Moreover, every worker receives only one $A(x_i)$. Thus, in case of $T$ colluding workers, the attackers can obtain at most $T$ evaluations of $A(x)$. Since $x_i$'s are distinct, only $T-1$ of $R_i$'s can be eliminated. Thus, $I(A;\{A(x_{i})\mid i\in \tilde{N}\})=0$ since $R_i$'s are generated uniform randomly from $\mathbb{F}_q$. Similarly, the attackers can obtain at most $mT$ evaluations of $B(x,y)$ and its $y$-directional derivatives in total. Since $(x_i,y_i)$'s are distinct for different $i$, and monomials $x^{K+i-1}y^{i-1}$ and their derivatives with respect to $y$ up to the $m^{th}$ order are linearly independent, the attackers can eliminate at most $mT-1$ of $S_{i,j}$'s. Thus, $I(B;\{\mathcal{B}_{i}\mid i\in \tilde{N}\})=0$ since $S_{i,j}$'s are generated uniform randomly from $\mathbb{F}_q$. 
\vspace{-5pt}
\end{proof}
\vspace{-5pt}
\subsection{Computation\label{subsec:Computation}}

Worker $i$ multiplies $A(x_{i})$ and $\partial_{j-1} B(x_{i},y_{i})$ with the increasing order of $j\in[m]$. That is, $j^{th}$ completed computation is $A(x_{i})\partial_{j-1} B(x_{i},y_{i})$. As soon as each multiplication
is completed, its result is communicated back to the master.
\vspace{-5pt}
\subsection{Decoding}


After collecting sufficiently many computations from the workers, the master can interpolate $A(x)B(x,y)$. Note that, in our scheme, every computation is equally useful, i.e., the sub-tasks are one-to-any replaceable. In  the following theorem, we give the recovery threshold expression, which specifies the minimum number of required computations and characterizes the probability of decoding failure, i.e., bivariate polynomial interpolation, due to the use of finite field. 

\begin{thm} \label{thm:r_th}
Assume the evaluation points $(x_i,y_i)$ are chosen uniform
randomly over the elements of $\mathbb{F}_{q}$. If the number of computations of sub-tasks received from the workers, which obey
the computation order specified in \subsecref{Computation} is greater than the recovery threshold $R_{th} \triangleq (K+T)L+m(K+T-1)$, then with probability at least $1-d/q$,
the master can interpolate the unique polynomial $A(x)B(x,y)$,
where
\begin{equation}
\label{eq:thm2}
d\triangleq \frac{m}{2}\left(3(K+T)^{2}+m(K+T)-8K-6T-m+3\right)+\frac{(K+T)L}{2}\left(K+L+T-2\right).
\end{equation}

\end{thm}

We give the proof sketch of \thmref{r_th} in the Appendix. \thmref{r_th} says that we can make the probability of failure arbitrarily small by increasing the order $q$ of the finite field.
\vspace{-5pt}
\section{Discussion}
\vspace{-5pt}
\label{sec:discussion}
The recovery threshold of our scheme is comparable to that of the multi-message extension of \cite{d2020gasp}, in which each worker is assigned $m$ computations. In this case, the number of evaluations of each encoding polynomial collectively obtained by $T$ colluding workers is $mT$. Thus, the recovery threshold of such an extension is obtained by substituting $mT$ for $T$ in the recovery threshold expression of \cite{d2020gasp}. That yields  \[
R_{th}^{GASP}=\begin{cases}
\begin{array}{lcc}
KL+K+L &  & 1=mT<L\leq K\\
KL+K+L+(mT)^{2}+mT-3 &  & 1<mT<L\leq K\\
\left(K+mT\right)(L+1)-1 &  & L\leq mT<K\\
2KL+2mT-1 &  & L\leq K\leq mT.
\end{array}\end{cases}
\]


In multi-message schemes, using univariate polynomial codes, including \cite{bitar2020rateless} and the multi-message extension of \cite{d2020gasp}, if a worker is assigned $m$ sub-tasks, then $m$ coded partitions of both $A$ and $B$ are required. Thus, the total upload cost for the univariate polynomial codes is $Nm\left(rs/K+sc/L\right)\log_2(q)$ bits. On the other hand, in the proposed scheme, we need one coded partition of $A$ and $m$ coded partitions of $B$. Thus, the upload cost of our scheme is $N\left(rs/K+msc/L\right)\log_2(q)$ bits, which is much less than that of univariate polynomial codes.

Comparing $R_{th}^{GASP}$ and $R_{th}$ of the proposed scheme for the same $m$
value might be misleading since, under the same upload cost budget, these schemes might have different $m$ values. Given an upload cost budget, we take the largest possible $m$ since it is the limiting factor for the  maximum number of computation a worker can provide. In \figref{r_th}, for $K=L=5,T=3$ we show how $m$ changes with the upload cost budget, and we compare $R_{th}$ and $R_{th}^{GASP}$ under a fixed upload cost budget, which is given in the number of matrix partitions for simplicity, assuming the partitions of $A$ and $B$ have the same size. Observe that in the proposed scheme, since $m$ cannot exceed $L=5$, for the upload cost values greater than 6, the values of $m$ and $R_{th}$ stays the same. 
\begin{figure}
    \centering
    \input{r_th_fig}
    \vspace{-10pt}
    \caption{Change of $m$ with upload cost budget and the comparison of the recovery thresholds of the proposed scheme and  the multi-message extension of \cite{d2020gasp}, when $K=L=5, T=3$.}
    \label{fig:r_th}
\end{figure}
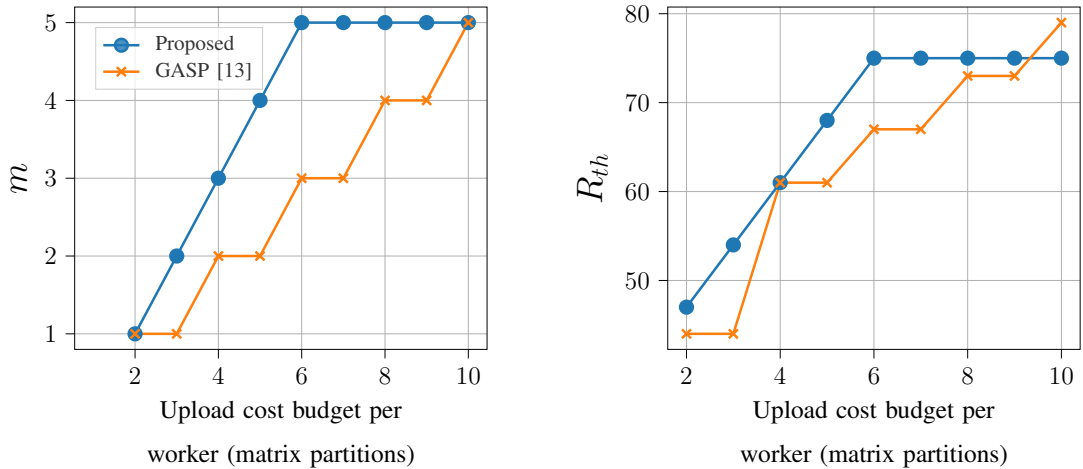

The upload cost could be further improved by employing the schemes in \cite{hasircioglu2020isit,hasircioglu_globecom_2020}, in which for a worker to complete $m$ computations, as few as $\sqrt{m}$ coded partitions of $A$ and $B$ may be enough. However, the extensions of these schemes to the private case have a large privacy overhead in the recovery threshold. For example, \cite{hasircioglu2020isit} originally has $R_{th}=KL$, but for its direct privacy extension, we have $R_{th}\approx (K+\sqrt{m}T)(L+\sqrt{m}T)$, which requires much larger computation time.
\vspace{-5pt}
\section{Simulation Results}
\vspace{-5pt}
In this section, we compare our work with the previous works in the literature \cite{d2020gasp,bitar2020rateless} in terms of average computation time. We define the \emph{computation time} as the time required by the workers to complete the number of computations specified by the recovery threshold. Our scenario consists of a limited number of workers and limited upload resources from the master to the workers. 

Following the literature \cite{liang2014tofec,lee2017speeding}, we assume that the time for a worker to finish one sub-task is distributed as a shifted exponential random variable with the density $\lambda e^{-\lambda(t-\nu)}$, where the scale parameter $\lambda$ controls the speed of the worker and the shift parameter $\nu$ is the minimum time duration for a task to be completed. Smaller $\lambda$ implies slower workers and thus more straggling. 

In our experiments, all the workers have a common shift parameter of $\nu=10/(KL)$ seconds. We assume $T=3$. We consider both matrices $A$ and $B$ are divided into $K=L=5$ partitions and assume that the partitions of matrices $A$ and $B$ have the same size, i.e., $\frac{r}{K}=\frac{c}{L}$. 

We consider two scenarios, workers with heterogeneous and homogeneous computational speeds, respectively. In the heterogeneous case, we group the workers into three classes, each consisting of 17 workers, and the computation speed of the workers in each class is specified by $\lambda_1=10^{-1}\times KL$, $\lambda_2=10^{-3}\times KL$ and $\lambda_3=10^{-4}\times KL$, respectively. In the homogeneous case, all the workers have $\lambda=10^{-2} \times KL$. Note that the coding scheme is agnostic to these values in both scenarios. 


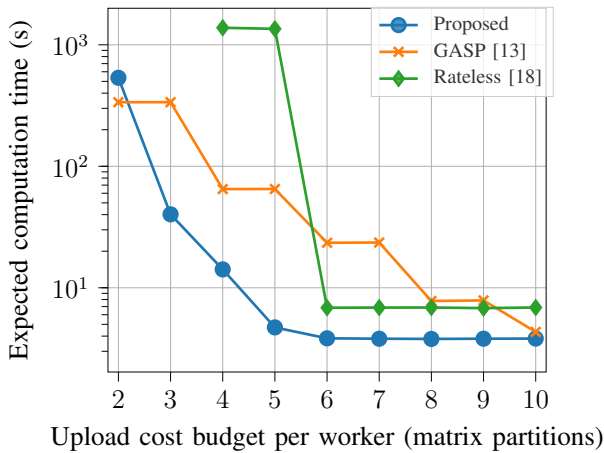
\begin{figure}
\vspace{-5pt}
    \centering
    \input{K-L-5.tex}
    \caption{Comparison of expected computation time as a function of the upload cost budget per worker when $K=L=5$ with heterogeneous workers.}
    \label{fig:heterogeneous}
    \vspace{-5pt}
\end{figure}

\begin{figure}
\vspace{-5pt}
    \centering
    \input{K-5_L-5-homogeneous}
    \caption{Comparison of expected computation time as a function of the upload cost budget per worker when $K=L=5$ with homogeneous workers.}
    \label{fig:homogeneous}
    \vspace{-5pt}
\end{figure}
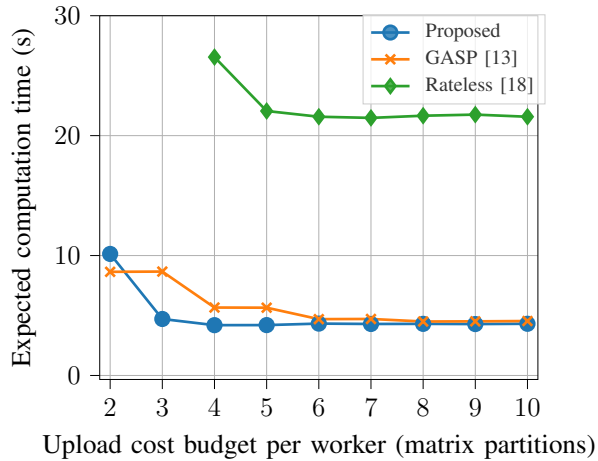

Since \cite{d2020gasp} was not originally proposed as a multi-message solution, we use its  multi-message extension as described in \secref{discussion}. For \cite{bitar2020rateless}, we do not limit the upload cost per worker but we limit the total upload cost since it allocates the computation load to workers adaptively in rounds. Moreover, since there are three groups of workers with different scale parameters, we use three clusters in the rounds after the first round. In the first round, we assign every worker one computation in accordance with the description in \cite{bitar2020rateless}.

In \figref{heterogeneous} and \figref{homogeneous}, we plot the expected computation time versus the upload cost per worker for heterogeneous and homogeneous cases, respectively. For simplicity, the upload cost is given in terms of the maximum number of total matrix partitions that can be sent to each worker, instead of the number of bits. As observed in \figref{heterogeneous}, in the heterogeneous case, even if the recovery threshold of the proposed scheme is close to that of GASP codes \cite{d2020gasp} (see \figref{r_th}), the proposed scheme's ability to generate more computations with the same upload cost budget, i.e. number of encoded matrix partitions, allows faster completion of the overall task. Rateless codes \cite{bitar2020rateless}, on the other hand, perform very close to the proposed scheme when the upload cost budget is large, but for the smaller values, either it could not complete, e.g., upload cost budgets 2 and 3, or it took too long to complete the task. This is because the scheme assigns new sub-tasks whenever a cluster completes its assignment greedily. Thus, the upload cost budget is mostly invested in the fastest cluster in the heterogeneous setting. However, despite its speed, the fastest cluster does not provide many useful computations compared to the slower clusters.   

On the other hand, for the homogeneous case, as observed in \figref{homogeneous}, the improvement of our scheme over GASP codes is limited. In this case, since workers' computation speeds are similar, faster workers do not compensate for slower ones. Still, we observe that due to the one-to-any replaceability of our scheme and GASP, both perform much better than rateless codes \cite{bitar2020rateless}. 

\vspace{-10pt}
\section{Conclusion}
\vspace{-5pt}

We have proposed storage- and upload-cost-efficient bivariate Hermitian polynomial codes for straggler exploitation in private distributed matrix multiplication. Previous works usually assume the availability of at least as many workers as the recovery threshold, but if the number of workers is not sufficient, the multi-message approach can allow the completion of the task. Compared to prior work, the proposed coding scheme has lower upload cost and less storage requirement, making the assignment of several sub-tasks to each worker more practical. Thanks to these properties, the proposed bivariate polynomial codes improve the average computation time of the private distributed matrix multiplication, especially when the number of workers, the upload cost budget or the storage capacity is limited.
\vspace{-10pt}
\renewcommand{\appendixname}{Appendix: Proof Sketch of \thmref{r_th}}
\appendix

In \figref{poly-coeffs}, we visualize the degrees of the monomials
of $A(x)B(x,y)$. We see that the number of monomials of $A(x)B(x,y)$ is $(K+T)L+m(K+T-1)$.
We need to show that every possible combination of so many responses from the workers interpolates to a unique polynomial, implying $(K+T)L+m(K+T-1)$ is the recovery threshold. 

\begin{figure}[tbh]
\vspace{-8pt}
\centering
\usetikzlibrary{patterns}
\begin{tikzpicture}[scale=1.7]
\draw[very thin]  (-5,2.5) rectangle (-2,0.5); 
\draw[very thin]  (-2,1.5) rectangle (0.5,0.5);
\path [pattern=north east lines, pattern color = black!20]  (-5,2.5) rectangle (-2,0.5); 
\path [pattern=crosshatch, pattern color = black!20]  (-2,1.5) rectangle (0.5,0.5);
\draw[fill, color=black] (-4.5,0.5) node (v1) {} circle (.05);
\draw[fill, color=black] (-4.5,2.5) node (v1) {} circle (.05);
\draw[fill, color=black] (-4.5,2) node (v1) {} circle (.05);
\draw[fill, color=black] (-4.5,1.5) node (v1) {} circle (.05);
\draw[fill, color=black] (-4.5,1) node (v1) {} circle (.05);
\draw[fill, color=black] (-5,0.5) node (v0) {} circle (.05); 
\draw[fill, color=black] (-5,2.5) node (v1) {} circle (.05);
\draw[fill, color=black] (-5,2) node (v1) {} circle (.05);
\draw[fill, color=black] (-5,1.5) node (v1) {} circle (.05);
\draw[fill, color=black] (-5,1) node (v1) {} circle (.05);
\draw[fill, color=black] (-4,0.5) node (v1) {} circle (.05);
\draw[fill, color=black] (-4,2.5) node (v1) {} circle (.05);
\draw[fill, color=black] (-4,2) node (v1) {} circle (.05);
\draw[fill, color=black] (-4,1.5) node (v1) {} circle (.05);
\draw[fill, color=black] (-4,1) node (v1) {} circle (.05);
\draw[fill, color=black] (-2,0.5) node (v1) {} circle (.05);
\draw[fill, color=black] (-2,2.5) node (v1) {} circle (.05);
\draw[fill, color=black] (-2,2) node (v1) {} circle (.05);
\draw[fill, color=black] (-2,1.5) node (v1) {} circle (.05);
\draw[fill, color=black] (-2,1) node (v1) {} circle (.05);
\draw[fill, color=black] (-2.5,0.5) node (v1) {} circle (.05);
\draw[fill, color=black] (-2.5,2.5) node (v1) {} circle (.05);
\draw[fill, color=black] (-2.5,2) node (v1) {} circle (.05);
\draw[fill, color=black] (-2.5,1.5) node (v1) {} circle (.05);
\draw[fill, color=black] (-2.5,1) node (v1) {} circle (.05);
\draw[fill, color=black] (-1.5,1.5) node (v1) {} circle (.05); 
\draw[fill, color=black] (-1.5,0.5) node (v1) {} circle (.05); 
\draw[fill, color=black] (-1.5,1) node (v1) {} circle (.05);
\draw[fill, color=black] (0,1.5) node (v1) {} circle (.05); 
\draw[fill, color=black] (0,0.5) node (v1) {} circle (.05); 
\draw[fill, color=black] (0,1) node (v1) {} circle (.05);
\draw[fill, color=black] (0.5,1.5) node (v1) {} circle (.05); 
\draw[fill, color=black] (0.5,0.5) node (v1) {} circle (.05); 
\draw[fill, color=black] (0.5,1) node (v1) {} circle (.05);
\node at (-0.75,1) {$\cdots$}; 
\node at (-3.25,1) {$\cdots$}; 
\node at (-3.25,1.5) {$\cdots$};
\node at (-3.25,2) {$\cdots$};
\node[scale=0.7] at (-5,0.25) {$0$}; 
\node[scale=0.7] at (-5.25,0.5) {$0$};
\node[scale=0.6] at (-2,0.25) {$K+T-1$};
\node[scale=0.7] at (-5.375,1.5) {$m-1$}; 
\node[scale=0.7] at (-5.375,2.5) {$L-1$}; 
\node[scale=0.6] at (0.5,0.25) {$2K+2T-2$};

\node[scale=0.8] (v2) at (1,0.5) {$\deg(x)$};
\draw [-latex] (v0) edge (v2);
\node[scale=0.8] (v3) at (-5,3) {$\deg(y)$};
\draw [-latex] (v0) edge (v3);

\end{tikzpicture}

\caption{\label{fig:poly-coeffs}The visualization of the degrees of the monomials
of $A(x)B(x,y)$ in the $\deg(x)-\deg(y)$ plane.}
\vspace{-10pt}
\end{figure}
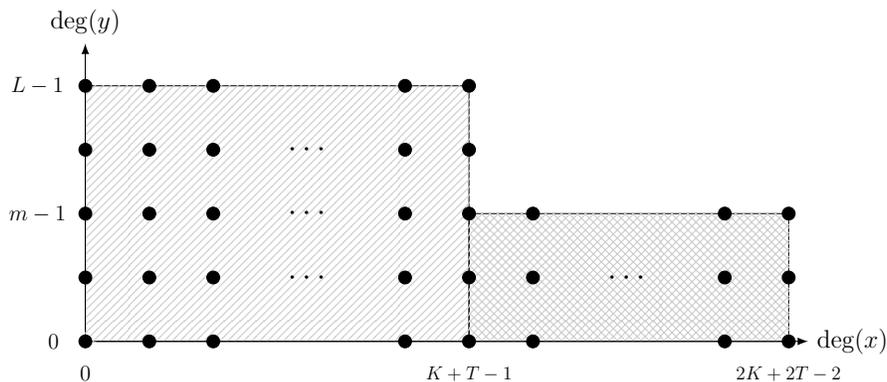
\begin{defn}
Bivariate polynomial interpolation problem can be formulated as solving
a linear system of equations, whose unknowns are the coefficients of $A(x)B(x,y)$ and whose coefficient matrix consists
of the monomials of $A(x)B(x,y)$ and their derivatives with respect to $y$, which are evaluated at the evaluation
points of responded workers. We refer to this coefficient matrix
as \textbf{interpolation matrix} and denote it by $M$. 
For example, when $K=L=2,m=2,T=1,N=4$, $M$ would be as follows.  $$
\begin{bmatrix}1 & x_{1} & x_{1}^{2} & x_{1}^{3} & x_{1}^{4} & y_{1} & x_{1}y_{1} & x_{1}^{2}y_{1} & x_{1}^{3}y_{1} & x_{1}^{4}y_{1}\\
0 & 0 & 0 & 0 & 0 & 1 & x_{1} & x_{1}^{2} & x_{1}^{3} & x_{1}^{4}\\
\vdots & \vdots & \vdots & \vdots & \vdots & \vdots & \vdots & \vdots & \vdots & \vdots\\
1 & x_{4} & x_{4}^{2} & x_{4}^{3} & x_{4}^{4} & y_{4} & x_{4}y_{4} & x_{4}^{2}y_{4} & x_{4}^{3}y_{4} & x_{4}^{4}y_{4}\\
0 & 0 & 0 & 0 & 0 & 1 & x_{4} & x_{4}^{2} & x_{4}^{3} & x_{4}^{4}
\end{bmatrix}$$
Observe that for worker 1, the first and the second rows correspond to $A(x)B(x,y)$ and $A(x)\partial_1B(x,y)$, evaluated at $(x_1,y_1)$, respectively.
\end{defn}

The problem of showing that any $R_{th}$ responses from the
workers interpolates to a unique polynomial is equivalent to showing
that the corresponding interpolation matrix is non-singular. The
theorem claims that this is the case with high probability.
First, we need to show that there exist some evaluation
points for which the determinant of the interpolation matrix is not
zero. That is equivalent to showing that $\det(M)$
is not the zero polynomial of the evaluation points. 

In \cite{hasircioglu2020bivariate} and \cite{hasircioglu_globecom_2020}, such a result for the same type of interpolation matrices is shown
for the real field $\mathbb{R}$. We omit it here for space restrictions. The result in \cite{hasircioglu2020bivariate} and \cite{hasircioglu_globecom_2020} is based on Taylor series expansion, which is also applicable in $\mathbb{F}_{q}$, if the degree of the polynomial $A(x)B(x,y)$ is smaller than $q$. This can be guaranteed by choosing a large $q$. For a discussion of the applicability of Taylor series expansion in finite fields, see \cite{hoffmanlinear}, and \cite{felix_fontein_2009}. From these results, we can conclude that $\det(M)$ is not the zero polynomial for large enough $q$. Next, we need to
find the upper bound on the probability $\det(M)=0$, when the evaluation points are sampled uniform
randomly from $\mathbb{F}_{q}$. 

\begin{lem}
\label{lem:schwartz_lemma}\textbf{\emph{\cite[Lemma 1]{schwartz1980fast}}}
Assume $P$ is a non-zero, $v$-variate polynomial of variables
$\alpha_{i},i\in[v]$. Let $d_{1}$ be the degree of $\alpha_{1}$
in $P(\alpha_{1},\dots,\alpha_{v})$, and $P_{2}(\alpha_{2},\dots,\alpha_{v})$
be the coefficient of $\alpha_{1}^{d_{1}}$in\textup{ $P(\alpha_{1},\dots,\alpha_{v})$.
}\textup{\emph{Inductively, let $d_{j}$ be the degree of $\alpha_{j}$
in $P_{j}(\alpha_{j},\dots,\alpha_v)$ and $P_{j+1}(\alpha_{j+1},\dots,\alpha_{v})$
be the coefficient of $\alpha_{j}$ in $P_{j}(\alpha_{j},\dots,\alpha_{v})$.
Let $S_{j}$ be a set of elements from a field $\mathbb{F}$, from
which the coefficients of $P$ are chosen. Then, in the Cartesian product set $S_{1}\times S_{2}\times\dots\times S_{v}$,
$P(\alpha_{1},\dots,\alpha_{v})$ has at most 
$
\left|S_{1}\times S_{2}\times\dots\times S_{v}\right|\left(\frac{d_{1}}{|S_{1}|}+\frac{d_{2}}{|S_{2}|}+\dots+\frac{d_{v}}{|S_{v}|}\right)
$
zeros.}}
\end{lem}
In our case, since the elements of $M$ are the monomials of $A(x)B(x,y)$ and their derivatives with respect to $y$, evaluated at some $(x_i,y_i)$, $\det(M)$ is a multivariate polynomial of the evaluation points $(x_i,y_i)$. Thus, $v$ is the number of different evaluation points in $M$. We choose the evaluation points from the whole field
$\mathbb{F}_{q}$. Thus, $S_{j}=\mathbb{F}_{q}$ and $|S_{j}|=q,\forall j\in[v]$,
and $\left|S_{1}\times S_{2}\times\dots\times S_{v}\right|=q^{v}$.
Then, the number of zeros of $\det(M)$ is at most $q^{v-1}(d_{1}+d_{2}+\dots+d_{v})$.
If we sample the evaluation points uniform randomly, then the probability
that $\det(M)=0$ is $(d_{1}+d_{2}+\dots+d_{v})/q$, since we sample
a $v$-tuple of evaluation points from $S_{1}\times S_{2}\times\dots\times S_{v}$. To find $d_{1}+d_{2}+\dots+d_{v}$,
we resort to the definition of determinant, that is $\det(M)=\sum_{i=1}^{R_{th}}(-1)^{1+i}m_{1,i}M_{1,i}$,
where $m_{1,i}$ is the element of $M$ at row 1 and column $i$ and
$M_{1,i}$ is the minor of $M$ when row 1 and column $i$ are removed \cite[Corollary 7.22]{liesen_linear_2015}. Thus, to identify the coefficients in \lemref{schwartz_lemma}, in the first row
of $M$, we start with the monomial with the largest degree. Assuming
the monomials are placed in an increasing order of their degrees,
the largest degree monomial is at column $R_{th}$. If that monomial
is univariate, then $d_{1}$ is the degree of the monomial and the
coefficient of $\alpha_{1}^{d_{1}}$ is $P_{2}(x_{2},\dots,x_{v})=\det(M_{1,1})$.
If the monomial is bivariate, then we take the degree of the corresponding
evaluation of $x$, i.e., $\alpha_{1}$, as $d_{1}$, and the degree
of the corresponding evaluation of $y$, i.e., $\alpha_{2}$, as $d_{2}$.
In this case, the coefficient of $\alpha^{d_{2}}$ is $P_{3}(\alpha_{3},\dots,\alpha_{v})=\det(M_{1,1})$.
Next, we take $M_{1,1}$, and repeat the same procedure. We do
so until we reach a monomial of degree zero.
In this procedure since we visit all the monomials of $A(x)B(x,y)$
evaluated at different evaluation points, i.e., $\alpha_{i}$'s, the
sum $d_{1}+d_{2}+\dots+d_{v}$ becomes the sum of degrees of all the
monomials of $A(x)B(x,y)$. The next lemma helps us in computing this.
\vspace{-5pt}
\begin{lem}
\label{lem:gauss-trick}Consider the polynomial $P(x,y)=\sum_{i=0}^{a}\sum_{j=0}^{b}c_{ij}x^{i}y^{j}$,
where $c_{i,j}$'s are scalars. The sum of degrees of all the monomials of
$P(x,y)$ is given by $\xi(a,b)\triangleq\frac{a(a+1)}{2}(b+1)+\frac{b(b+1)}{2}(a+1)$. 
\end{lem}
The proof of \lemref{gauss-trick} is based on Gauss's trick,
and is omitted due to space restrictions. By using \lemref{gauss-trick}, we can find the sum of monomial degrees in the diagonally shaded rectangle and the rectangle shaded by crosshatches in \figref{poly-coeffs}, separately, and by summing them we find $d_{1}+d_{2}+\dots+d_{v}$ to be equal to \eqref{thm2}, which concludes the proof. 


\newpage

\bibliographystyle{IEEEtran}
\bibliography{refs}

\end{document}

%% file: r_th_fig.tex
\begin{tikzpicture}[scale=0.8]

\definecolor{color0}{rgb}{0.12156862745098,0.466666666666667,0.705882352941177}
\definecolor{color1}{rgb}{1,0.498039215686275,0.0549019607843137}

\begin{groupplot}[group style={group size=2 by 2, horizontal sep=3cm}]
\nextgroupplot[
legend cell align={left},
legend style={fill opacity=0.8, draw opacity=1, text opacity=1, draw=white!80!black,
at={(0.05,0.85)},anchor=west},
tick align=outside,
tick pos=left,
x grid style={white!69.0196078431373!black},
xlabel style={align=center, text width=8cm},
xlabel={Upload cost budget per worker (matrix partitions)},
xmin=0.55, xmax=10.45,
xtick style={color=black},
y grid style={white!69.0196078431373!black},
ylabel={\Large \(\displaystyle m\)},
ymin=0.8, ymax=5.2,
ytick style={color=black},
grid
]
\addplot [very thick, color0, mark=*, mark size=3, mark options={solid}]
table {%
2 1
3 2
4 3
5 4
6 5
7 5
8 5
9 5
10 5
};
\addlegendentry{\small \hspace{-7pt} Proposed}

\addplot [very thick, color1, mark=x, mark size=3, mark options={solid}]
table {%
2 1
3 1
4 2
5 2
6 3
7 3
8 4
9 4
10 5
};
\addlegendentry{\small \hspace{-7pt} GASP \cite{d2020gasp}}

\nextgroupplot[
tick align=outside,
tick pos=left,
x grid style={white!69.0196078431373!black},
xlabel style={align=center, text width=8cm},
xlabel={Upload cost budget per worker (matrix partitions)},
xmin=1.6, xmax=10.4,
xtick style={color=black},
y grid style={white!69.0196078431373!black},
ylabel={\Large  \(\displaystyle R_{th}\)},
ymin=42.25, ymax=80.75,
ytick style={color=black},
grid
]
\addplot [very thick, color0, mark=*, mark size=3, mark options={solid}]
table {%
2 47
3 54
4 61
5 68
6 75
7 75
8 75
9 75
10 75
};
\addplot [very thick, color1, mark=x, mark size=3, mark options={solid}]
table {%
2 44
3 44
4 61
5 61
6 67
7 67
8 73
9 73
10 79
};
\end{groupplot}

\end{tikzpicture}

%% file: K-L-5.tex

%
%
%
%
%

\begin{tikzpicture}[scale=0.85]

\definecolor{color0}{rgb}{0.12156862745098,0.466666666666667,0.705882352941177}
\definecolor{color1}{rgb}{1,0.498039215686275,0.0549019607843137}
\definecolor{color2}{rgb}{0.172549019607843,0.627450980392157,0.172549019607843}
\definecolor{color3}{rgb}{0.83921568627451,0.152941176470588,0.156862745098039}

\begin{semilogyaxis}[
legend cell align={left},
legend style={fill opacity=0.8, draw opacity=1, text opacity=1, draw=white!80!black,
at={(0.60,0.88)},anchor=west},
tick align=outside,
tick pos=left,
x grid style={white!69.0196078431373!black},
xlabel={Upload cost budget per worker (matrix partitions)},
xmin=1.8, xmax=10.2,
xtick style={color=black},
xtick={2,3,4,5,6,7,8,9,10},
y grid style={white!69.0196078431373!black},
ylabel={Expected computation time (s)},
ymin=-0.333274999999998, ymax=2000,
ytick style={color=black},
grid
]
\addplot [very thick, color0, mark=*, mark size=3, mark options={solid}]
table {%
2 535.398
3 40.233
4 14.168
5 4.717
6 3.836
7 3.809
8 3.792
9 3.809
10 3.817
};
\addlegendentry{\footnotesize \hspace{-7pt} Proposed}

\addplot [very thick, color1, mark=x, mark size=3, mark options={solid}]
table {%
2  336.89
3  337.181
4  64.92
5  65.046
6  23.473
7  23.578
8  7.782
9  7.871
10 4.331
};
\addlegendentry{\footnotesize \hspace{-7pt} GASP \cite{d2020gasp}}


\addplot [very thick, color2, mark=diamond*, mark size=3, mark options={solid, fill}]
table {%
4  1381.37  
5  1353.58  
6  6.85     
7  6.87    
8  6.88    
9  6.80    
10 6.88    
};
\addlegendentry{\footnotesize \hspace{-7pt} Rateless \cite{bitar2020rateless}}


\end{semilogyaxis}

\end{tikzpicture}

%% file: K-5_L-5-homogeneous.tex

%
%
%
%
%

\begin{tikzpicture}[scale=0.85]

\definecolor{color0}{rgb}{0.12156862745098,0.466666666666667,0.705882352941177}
\definecolor{color1}{rgb}{1,0.498039215686275,0.0549019607843137}
\definecolor{color2}{rgb}{0.172549019607843,0.627450980392157,0.172549019607843}
\definecolor{color3}{rgb}{0.83921568627451,0.152941176470588,0.156862745098039}

\begin{axis}[
legend cell align={left},
legend style={fill opacity=0.8, draw opacity=1, text opacity=1, draw=white!80!black,
at={(0.60,0.88)},anchor=west},
tick align=outside,
tick pos=left,
x grid style={white!69.0196078431373!black},
xlabel={Upload cost budget per worker (matrix partitions)},
xmin=1.8, xmax=10.2,
xtick style={color=black},
xtick={2,3,4,5,6,7,8,9,10},
y grid style={white!69.0196078431373!black},
ylabel={Expected computation time (s)},
ymin=-0.333274999999998, ymax=30,
ytick style={color=black},
grid
]
\addplot [very thick, color0, mark=*, mark size=3, mark options={solid}]
table {%

2 10.135
3 4.711
4 4.195
5 4.201
6 4.326
7 4.295
8 4.304
9 4.287
10 4.311
};
\addlegendentry{\footnotesize \hspace{-7pt} Proposed}
\addplot [very thick, color1, mark=x, mark size=3, mark options={solid}]
table {%

2  8.643
3  8.664
4  5.665
5  5.652
6  4.697
7  4.712
8  4.5
9  4.513
10 4.54
};
\addlegendentry{\footnotesize \hspace{-7pt} GASP \cite{d2020gasp}}


\addplot [very thick, color2, mark=diamond*, mark size=3, mark options={solid, fill}]
table {%
4  26.549
5  22.054
6  21.576
7  21.479
8  21.661
9  21.753
10 21.575
};
\addlegendentry{\footnotesize \hspace{-7pt} Rateless \cite{bitar2020rateless}}


\end{axis}

\end{tikzpicture}

%% file: template_isit20.bbl
\begin{thebibliography}{10}
\providecommand{\url}[1]{#1}
\csname url@samestyle\endcsname
\providecommand{\newblock}{\relax}
\providecommand{\bibinfo}[2]{#2}
\providecommand{\BIBentrySTDinterwordspacing}{\spaceskip=0pt\relax}
\providecommand{\BIBentryALTinterwordstretchfactor}{4}
\providecommand{\BIBentryALTinterwordspacing}{\spaceskip=\fontdimen2\font plus
\BIBentryALTinterwordstretchfactor\fontdimen3\font minus
  \fontdimen4\font\relax}
\providecommand{\BIBforeignlanguage}[2]{{%
\expandafter\ifx\csname l@#1\endcsname\relax
\typeout{** WARNING: IEEEtran.bst: No hyphenation pattern has been}%
\typeout{** loaded for the language `#1'. Using the pattern for}%
\typeout{** the default language instead.}%
\else
\language=\csname l@#1\endcsname
\fi
#2}}
\providecommand{\BIBdecl}{\relax}
\BIBdecl

\bibitem{lee2017speeding}
K.~Lee, M.~Lam, R.~Pedarsani, D.~Papailiopoulos, and K.~Ramchandran, ``Speeding
  up distributed machine learning using codes,'' \emph{IEEE Transactions on
  Information Theory}, vol.~64, no.~3, pp. 1514--1529, 2017.

\bibitem{yu_polynomial_2017}
Q.~Yu, M.~Maddah-Ali, and S.~Avestimehr, ``Polynomial codes: an optimal design
  for high-dimensional coded matrix multiplication,'' in \emph{Advances in
  Neural Information Processing Systems}, 2017, pp. 4403--4413.

\bibitem{dutta_optimal_2019}
S.~Dutta, M.~Fahim, F.~Haddadpour, H.~Jeong, V.~Cadambe, and P.~Grover, ``On
  the optimal recovery threshold of coded matrix multiplication,'' \emph{IEEE
  Transactions on Information Theory}, vol.~66, no.~1, pp. 278--301, 2019.

\bibitem{yu_straggler_2018-1}
Q.~Yu, M.~A. Maddah-Ali, and A.~S. Avestimehr, ``Straggler mitigation in
  distributed matrix multiplication: Fundamental limits and optimal coding,''
  in \emph{IEEE International Symposium on Information Theory}, 2018.

\bibitem{jia2019cross}
Z.~Jia and S.~A. Jafar, ``Cross subspace alignment codes for coded distributed
  batch computation,'' \emph{arXiv}, pp. arXiv--1909, 2019.

\bibitem{kiani_exploitation_2018}
S.~Kiani, N.~Ferdinand, and S.~C. Draper, ``Exploitation of stragglers in coded
  computation,'' in \emph{IEEE International Symposium on Information Theory},
  2018.

\bibitem{amiri_computation_2018}
M.~M. Amiri and D.~G{\"u}nd{\"u}z, ``Computation scheduling for distributed
  machine learning with straggling workers,'' \emph{IEEE Transactions on Signal
  Processing}, vol.~67, no.~24, pp. 6270--6284, 2019.

\bibitem{ozfatura2020straggler}
E.~Ozfatura, S.~Ulukus, and D.~G{\"u}nd{\"u}z, ``Straggler-aware distributed
  learning: Communication--computation latency trade-off,'' \emph{Entropy},
  vol.~22, no.~5, p. 544, 2020.

\bibitem{hasircioglu2020isit}
B.~Hasircioglu, J.~Gomez-Vilardebo, and D.~Gunduz, ``Bivariate polynomial
  coding for straggler exploitation with heterogeneous workers,'' \emph{IEEE
  International Symposium on Information Theory}, 2020.

\bibitem{hasircioglu_globecom_2020}
B.~Hasircioglu, J.~Gomez-Vilardebo, and D.~Gunduz, ``Bivariate hermitian
  polynomial coding for efficient distributed matrix multiplication,'' in
  \emph{IEEE Global Communications Conference (GLOBECOM)}, 2020.

\bibitem{chang2018capacity}
W.-T. Chang and R.~Tandon, ``On the capacity of secure distributed matrix
  multiplication,'' in \emph{2018 IEEE Global Communications Conference
  (GLOBECOM)}.\hskip 1em plus 0.5em minus 0.4em\relax IEEE, 2018, pp. 1--6.

\bibitem{kakar2018rate}
J.~Kakar, S.~Ebadifar, and A.~Sezgin, ``Rate-efficiency and
  straggler-robustness through partition in distributed two-sided secure matrix
  computation,'' \emph{arXiv preprint arXiv:1810.13006}, 2018.

\bibitem{d2020gasp}
R.~G.~L. {D’Oliveira}, S.~{El Rouayheb}, and D.~{Karpuk}, ``{GASP} codes for
  secure distributed matrix multiplication,'' \emph{IEEE Transactions on
  Information Theory}, vol.~66, no.~7, pp. 4038--4050, 2020.

\bibitem{aliasgari2020private}
M.~Aliasgari, O.~Simeone, and J.~Kliewer, ``Private and secure distributed
  matrix multiplication with flexible communication load,'' \emph{IEEE
  Transactions on Information Forensics and Security}, vol.~15, pp. 2722--2734,
  2020.

\bibitem{jia2019capacity}
Z.~Jia and S.~A. Jafar, ``On the capacity of secure distributed matrix
  multiplication,'' \emph{arXiv preprint arXiv:1908.06957}, 2019.

\bibitem{kakar2019uplink}
J.~Kakar, A.~Khristoforov, S.~Ebadifar, and A.~Sezgin, ``Uplink-downlink
  tradeoff in secure distributed matrix multiplication,'' \emph{arXiv preprint
  arXiv:1910.13849}, 2019.

\bibitem{mital2020secure}
N.~Mital, C.~Ling, and D.~Gunduz, ``Secure distributed matrix computation with
  discrete fourier transform,'' \emph{arXiv preprint arXiv:2007.03972}, 2020.

\bibitem{bitar2020rateless}
R.~Bitar, M.~Xhemrishi, and A.~Wachter-Zeh, ``Rateless codes for private
  distributed matrix-matrix multiplication,'' \emph{arXiv preprint
  arXiv:2004.12925}, 2020.

\bibitem{liang2014tofec}
G.~Liang and U.~C. Kozat, ``Tofec: Achieving optimal throughput-delay trade-off
  of cloud storage using erasure codes,'' in \emph{IEEE INFOCOM 2014-IEEE
  Conference on Computer Communications}.\hskip 1em plus 0.5em minus
  0.4em\relax IEEE, 2014, pp. 826--834.

\bibitem{hasircioglu2020bivariate}
B.~Hasircioglu, J.~Gomez-Vilardebo, and D.~Gunduz, ``Bivariate polynomial
  coding for exploiting stragglers in heterogeneous coded computing systems,''
  \emph{arXiv preprint arXiv:2001.07227}, 2020.

\bibitem{hoffmanlinear}
K.~Hoffman and R.~Kunze, ``Linear algebra,'' \emph{Englewood Cliffs, New
  Jersey}, 1971.

\bibitem{felix_fontein_2009}
\BIBentryALTinterwordspacing
F.~Fontein, ``The hasse derivative,'' Aug 2009. [Online]. Available:
  \url{https://math.fontein.de/2009/08/12/the-hasse-derivative/}
\BIBentrySTDinterwordspacing

\bibitem{schwartz1980fast}
J.~T. Schwartz, ``Fast probabilistic algorithms for verification of polynomial
  identities,'' \emph{Journal of the ACM (JACM)}, vol.~27, no.~4, pp. 701--717,
  1980.

\bibitem{liesen_linear_2015}
J.~Liesen and V.~Mehrmann, \emph{Linear algebra}, ser. Springer undergraduate
  mathematics series.

\end{thebibliography}
